\documentclass{article}

 \newcommand{\R}{\mathcal{R}}
\newcommand{\PA}{\mathcal{P}}

\newcommand{\N}{\mathcal{N}}
\newcommand{\ignore}[1]{}

\usepackage{blindtext} 
\usepackage{amsmath,amsfonts,graphicx,amssymb}
\usepackage{bm}
\usepackage{float}
\usepackage{amsthm}
\usepackage{mathtools}
\newtheorem{theorem}{Theorem}[section]
\newtheorem{lemma}{Lemma}[section]

\usepackage{color}
\title{Whittle Indexability in Egalitarian Processor Sharing Systems}

\author{Vivek S. Borkar \footnote{Vivek S. Borkar \newline Department of Electrical Engineering,\newline
	Indian Institute of Technology Bombay,\newline
	Powai, Mumbai 400076, India.\newline
              E-mail: borkar.vs@gmail.com \newline
              Research supported in part by CEFIPRA project 5100-IT1}        \and
        Sarath Pattathil \footnote{Sarath Pattathil \newline Department of Electrical Engineering,\newline
	Indian Institute of Technology Bombay,\newline
	Powai, Mumbai 400076, India.\newline
	E-mail: sarathpattathil@iitb.ac.in \newline} }

\date{}

\begin{document}




\maketitle
\begin{abstract}
The egalitarian processor sharing model is viewed as a restless bandit and its Whittle indexability is established. A numerical scheme for computing the Whittle indices is provided, along with supporting numerical experiments.
\end{abstract}

\section{Introduction}
A variety of processor sharing models have been proposed and analyzed over the last half century \cite{Altman2}, \cite{Kleinrock}. Numerous techniques have been used for this purpose. Our aim here is to take up one of the simplest such models and look at it from a different perspective. Specifically, we map it to a restless bandit problem and establish its Whittle indexability. While restless bandits has been a popular paradigm for resource allocation in queuing problems (see, e.g., \cite{Glaze1},  \cite{Glaze4},  \cite{Glaze2},  \cite{Glaze3},  \cite{Jacko},  \cite{Nino1},  \cite{Nino2},  \cite{Nino3}), its application to processor sharing appears new.  This application has some non-standard features such as possibly more than one customer being served by a server at the same time and scaling down of the service rate with the number of customers, which complicates the analysis significantly. We also propose a numerical scheme for approximate computation of Whittle indices and present some numerical experiments.\\

\noindent The article is structured as follows. The next section describes our model and the control problem formulation, its Whittle relaxation and its decoupling via Lagrange multiplier into agent-wise decoupled control problems. Section 3 derives the corresponding dynamic programming equation for the decoupled problem. Section 4 establishes various structural properties of the associated value functions, leading to the existence of an optimal threshold policy. This is the basis of the Whittle indexability proof which, along with a scheme for computing it, is presented in section 5. Section 6 sets up the stage for evaluating and comparing the performance of the Whittle index policy with other heuristics, with annotated numerical experiments being presented in section 7. Section 8 concludes with some comments.\\

\section{The Model}
We consider a simple processor sharing model. Specifically, consider a single Bernoulli arrival process $\{\xi_t\}$ with arrival probability $p \in (0, 1)$. Each arrival is directed to one of $I$ queues. The dynamics of the $i^{th}$ queue is as follows. Its queue length $X^i_t$ at time $t$ updates according to
\begin{equation}\label{eq:dynamics}
X^i_{t+1} = X^i_t - D^i_{t+1} + \nu^i_t\xi_{t+1}
\end{equation}
where $D^i_{t+1}$ is the departure process  as per the egalitarian processor sharing model described below and $\nu^i_t$ is $\{0, 1\}-$valued control variable satisfying: 1 denotes active state (i.e., arrivals are admitted) and 0 is passive (i.e., arrivals are not admitted). The latter is required to be conditionally independent of $D_s^i, \xi_s^i,\ 1 \leq i \leq I, \  s > t$, given  $\{X^i_0; \ D^i_s, \xi^i_s, \ s \leq t; \ \nu^i_s, \ s < t\} \ \forall t$. Such $\{\nu^i_t\}$ will be called `admissible controls'.
\newline
\newline
Egalitarian processor sharing is a service policy where all the customers being served are given the same processing power. This means that the departure process from server $i$ will be a binomial process, with the following properties. If the capacity is $q_i$, it means that if there is only one job in the server, the probability of it being completed in a given time slot is $q_i$. However, if there are K jobs in the server, by the egalitarian scheme, the probability of departure will be reduced to $\frac{q_i}{K}$ for each of the jobs in the server. Note that the mean number of departures in each time slot remains fixed at $q_i$. The departure from the server will be a binomial random variable with departure probability $\frac{q_i}{K}$, i.e.,
\begin{equation}
\mathbb{P}\text{(d departures when x jobs in server)} = {x \choose d}\bigg( \frac{q_i}{x} \bigg)^d \bigg( 1- \frac{q_i}{x} \bigg) ^{x-d}.
\label{Dep_Process}
\end{equation}
The cost of having a job in server $i$ is given by $C_i > 0$.
\newline
\newline
Our aim is to minimize the long run average cost
$$\limsup_{T\uparrow\infty}\frac{1}{T}\sum_{t=0}^{T-1}\sum_i C_iE[X^i_t(\nu^i_t)].$$
subject to
$$\sum_i\nu^i_t = 1 \ \forall m.$$
The hard per stage constraint above makes the problem difficult \cite{Papa}. Following Whittle, we replace the hard constraint by its relaxation
$$\limsup_{T\uparrow\infty}\frac{1}{T}\sum_{t=0}^{T-1}\sum_iE[\nu^i_t] = 1.$$
A standard Lagrange multiplier formulation \cite{Borkar} leads to the unconstrained problem of minimizing
\begin{equation}
\limsup_{T\uparrow\infty}\frac{1}{T}\sum_{t=0}^{T-1}\sum_iE[c(X^i_t, \nu^i_t)] \label{ergcost}
\end{equation}
for
$$c(x,\nu) = Cx + (1-\nu)\lambda$$
where $\lambda$ is the Lagrange multiplier. We choose this specific form in view of the interpretation of $\lambda$ as negative subsidy in the context of Whittle indexability.\\

Note that this is now an average cost Markov decision process with a separable cost function. Hence, given the Lagrange multiplier, it decouples into individual control problems for the agents. This is what we exploit in what follows.\\

\noindent \textbf{Remark 1:} More generally, one may consider a general cost function of the form:
$c(x,\nu) = F(x) + (1-\nu)\lambda$ where $F$ is convex increasing. \\

\noindent Given $\lambda$, the problem gets decoupled into an individual control problem for each processor separately. Following Whittle, we view $\lambda$ as a `negative subsidy' or `tax' for passivity (negative because this is a minimization problem). The problem is Whittle indexable if for any choice of problem parameters $(C_i, q_i, p) \in (0, \infty)\times(0,1)\times(0,1)$, for each of these individual problems the set of passive states increases monotonically  from empty set to the whole state space as $\lambda$ decreases from $\infty$ to $-\infty$. Then for each individual problem, the Whittle index for state $x$ is the value of $\lambda$ for which both active and passive modes are equally preferable for $x$. The corresponding index policy is: for a given state profile across the processors, deem the one with the lowest index active, i.e., arrivals are permitted only for that particular processor.
\newline
\newline
Without loss of generality, the processing rates of different servers may be ordered as
\begin{align}
1>q_1 > q_2 > q_3 > ... >q_I>2p>0. \label{ordered}
\end{align}
Thus no server has departure probability 1. Also, $p < q_i \ \forall i$ ensures stability of the queues as we see below. 
\footnote{Note that for finite systems i.e. when the number of jobs in any server is upper bounded so that any additional jobs that arrive are dropped, this assumption is not required. We will use this fact in simulations.}
\newline
\newline
Given that the cost $C_i > 0 \ \forall i$, if $q_I \leq p$, the choice $\nu^I(n) \equiv 1 \ \forall n$, i.e., all jobs are always sent to server $I$,  leads to the stationary expectation of the cost of $X^I_t$ being $+\infty$, leading to an infinite cost. Thus it is  convenient to assume that $q_I > p$, implying $q_i > p \ \forall i$. We have gone for a stronger condition $q_I > 2p$. We further comment on this assumption in the concluding section.

\begin{lemma}
The function $\Phi([ x_1, \cdots, x_I]) = \sum_{i=1}^{I}e^{a x_i}$, where $a>0$ such that
\begin{align}
p(e^a - 1) < \frac{q_I}{2}(1 - e^{-a}) \nonumber
\end{align}
serves as the Lyapunov function, satisfying, for $X_t := [X^1_t, \cdots, X^I_t]$ and $\theta := [0, \cdots, 0]$,
$$E\left[\Phi(X_{t+1}) - \Phi(X_t) | X_t\right] \leq -b\Phi(X_t) \ \mbox{for} \ X_t \neq \theta$$
for some $b > 0$.
\end{lemma}
\begin{proof}
The probability that no jobs leave from queue $i$ is given by:
\begin{align}
\mathbb{P}(\text{0 departures from queue i}) &= \bigg(1 - \frac{q_i}{x_i} \bigg)^{x_i} \nonumber \\
&\leq e^{-q_i} \nonumber \\
&\leq e^{-q_I} \nonumber
\end{align}
This gives:
\begin{align}
\mathbb{P}(\text{at least 1 departure from queue i}) & \geq (1-e^{-q_I}) \nonumber \\
& \geq \frac{q_I}{2} \nonumber
\end{align}
Let
\begin{align}
b = \frac{q_I}{2}(1 - e^{-a}) - p(e^a - 1).\nonumber
\end{align}
Then $b > 0$ by our choice of $a$.
We have:
\begin{align}
E\left[\Phi(X_{t+1}) | \mathcal{F}_t\right] - \Phi(X_t)  &\leq \frac{q_I}{2}(e^{-a} - 1)\Phi(X_t) + p(e^a - 1)\Phi(X_t) \nonumber \\
&= \big(-\frac{q_I}{2}(1 - e^{-a}) + p(e^a - 1) \big)\Phi(X_t) \nonumber \\
&= -b \Phi(X_t) \nonumber \\
& < 0.
\label{stable}
\end{align}
This inequality is true regardless of the choice of $\{\nu^i_t\}$. The claim follows.
\end{proof}

We shall be specifically interested in the stationary policies, that is, $\bar{\nu}_t = [\nu^1_t, \cdots, \nu^I_t]$ of the form $\nu^i_t = v^i(X_t) \ \forall \ i,t$ for some $v^i : \N := \{0, 1, 2, \cdots\} \mapsto \{0, 1\}$ and $1 \leq i \leq I$. Under such policies, $\{X_t\}$ is a time-homogeneous Markov chain. By standard abuse of notation, we shall identify such a policy with the map $v(\cdot) := [v^1(\cdot), \cdots, v^I(\cdot)] : \N^I \mapsto$ the unit coordinate vectors in $\R^I$. From (\ref{stable}), we then have the following: Let $\PA(\N) :=$ the space of probability measures on $\N$ with Prohorov topology.

\begin{lemma}\label{general}
Under any stationary policy, $\{X_t\}$  has a single aperiodic communicating class $\widetilde{S}$ that includes $\theta$, possibly with some transient states. Restricted to $\widetilde{S}$, the chain is geometrically ergodic. Furthermore, if $\pi$ denotes its unique stationary distribution and $\tau_{\theta} := \{n \geq 0: X_t = \theta\}$  the first hitting time of $\theta$, then:
\begin{enumerate}

\item there exists $a > 0$ such that $\forall \ x \neq \theta$, $E\left[e^{a\tau_{\theta}}|X_0 = x\right] \leq K_x < \infty$ for some $K_x < \infty$, uniformly in $v$,

\item $\forall \ f : \mathcal{N} \mapsto \R$ that are $O(\Phi(\cdot))$, $E\left[f(X_t)\right] \to \sum_x\pi(x)f(x)$ exponentially, uniformly in $v$, in fact, for $x = [x_1, x_2, \cdots]$,
\begin{equation}
\left|E\left[f(X_t) | X_0 = x\right] - \sum_y\pi(y)f(y)\right| \leq K(1 + \sum_i e^{ax_i})\eta^n \label{expdecay}
\end{equation}
for some $K > 0$ and $0 < \eta < 1$,

\item $\sup E\left[\sum_iX^i_t\right] < \infty$, where the supremum is over all stationary policies and the expectation is over the corresponding stationary distributions. In particular, the latter are compact in $\PA(\N)$,

\item $\sup E\left[\sum_{t=0}^{\tau_{\theta}}X_t\right] < \infty$, with the supremum as above.
\end{enumerate}
\end{lemma}
\begin{proof} Assume for the time being that the chain is irreducible, i.e., $\widetilde{S} = \mathcal{N}$.
Fix $v$.   For a stable stationary policy, there must exist a positive recurrent state, and since there is a positive probability of moving to state $\theta$ from this state, it follows that $\theta$ is also positive recurrent. Aperiodicity follows from the observation that there is a nonzero probability of remaining in state $\theta$.

 In view of (\ref{stable}), it follows from Theorem 16.0.1, p.\ 393, \cite{Meyn}, that $\{X_t\}$ is in fact $\Phi$-uniformly ergodic in the sense of (16.2), p.\ 392, \cite{Meyn}, in particular, uniformly ergodic in the sense of (16.6), p.\ 393, \cite{Meyn} \textit{and} geometrically ergodic. The main claim and the first  bullet follow by specializing Theorem 16.0.2, p.\ 394, of \cite{Meyn} to the present case, the second follows likewise from Theorem 16.0.1, pp.\ 393, \cite{Meyn}.
The uniformity in $v$ follows in view of the common Lyapunov condition (\ref{stable}). The last two points follow as in Theorem 8.1, p.\ 108, \cite{Borkar_2}\footnote{In fact the claim/proof of the last claim extends to general admissible controls, though we won't need it here.}.

 Now suppose $x$ is transient. Given any state in the Markov Chain, there is a non-zero probability of moving to state $\theta$ from this state (the unichain property). By a standard argument using the stochastic Lyapunov condition (\ref{stable}) and the optional sampling theorem coupled with Fatou's lemma, one can show that
if $\zeta$ denotes the first hitting time of $\widetilde{S}$, then  $E\left[\Phi(X_{\xi})\right] \leq \Phi(x)$.
Therefore
\begin{eqnarray*}
\lefteqn{\lim_{t\uparrow\infty}\frac{1}{t}\log\left|E\left[f(X_t) | X_0 = x\right] - \sum_y\pi(y)f(y)\right| =} \\
&=& \lim_{t\uparrow\infty}\frac{1}{t}\log\left|E\left[E\left[f(X_{t}) | X_{\zeta}\right]X_0 = x\right] - \sum_y\pi(y)f(y)\right| \\
&& \ \ \ \ \ \mbox{by the strong Markov property} \\
&\leq& \lim_{t\uparrow\infty}\frac{1}{t}\log\left|E\left[K\Phi(X_{\zeta})\eta^t | X_0 = x\right]\right| \\
&\leq& \log(\eta)
\end{eqnarray*}
for a suitably redefined $K$. This extends the claim to the transient states.
\end{proof}

\section{The dynamic programming equation}

Recall our problem of controlling an individual chain to minimize (\ref{ergcost}). We drop the agent index $i$ for the time being for notational ease.  The dynamic programming equation for this problem is then given by
\begin{eqnarray}
\lefteqn{V(x)  = Cx - \beta \ + } \nonumber \\
&&\min \bigg( \sum_{d=0}^{x} {x \choose d} \bigg(\frac{q_i}{x} \bigg)^d \bigg(1-\frac{q_i}{x} \bigg)^{x-d}(pV(x+1-d)
+ (1-p)V(x-d)), \nonumber \\
&&\quad \qquad \lambda + \sum_{d=0}^{x} {x \choose d} \bigg(\frac{q_i}{x} \bigg)^d \bigg(1-\frac{q_i}{x} \bigg)^{x-d}V(x-d) \bigg).
\label{DP2}
\end{eqnarray}
Let $0<\alpha < 1$. The cost for the infinite horizon $\alpha$ discounted problem is:
\begin{align*}
J_\alpha(x,\{ \nu_t \}) := \mathbb{E} \bigg[ \sum_{t=0}^{\infty}\alpha^t(CX_t + (1-\nu_t)\lambda)\Big|X_0 = x \bigg].
\end{align*}
The value function corresponding to this discounted problem will be:
\begin{align*}
V_\alpha(x) = \min_{\nu \in \{ 0,1 \}}J_\alpha(x, \{ \nu _t \})
\end{align*}
where the minimum is over all admissible controls. The dynamic programming equation satisfied by this value function is:
\begin{align*}
V_\alpha (x) = \min_{\nu} \bigg[ Cx + (1-\nu)\lambda + \alpha \sum_{y} p_{\nu}(y|x)V_\alpha(y) \bigg]
\end{align*}
where $p_{\nu}( \cdot | \cdot )$ is the controlled transition probability. Then $\overline{V}_\alpha(.) = V_\alpha(.) - V_\alpha(0)$ satisfies:
\begin{align}
\overline{V}_\alpha (x) = \min_{\nu} \bigg[ Cx + (1-\nu)\lambda - (1-\alpha)V_\alpha(0) + \alpha \sum_{y} p_{\nu}(y|x)\overline{V}_\alpha(y) \bigg].
\label{DPbar}
\end{align}

\noindent Now consider the problem of minimizing (\ref{ergcost}). \\

\begin{lemma}
\label{lemma:Lyapunov}
There exists an optimal stable stationary policy $v^* : \N \mapsto \{0, 1\}$ and $W(x) = e^{ax}$, where a is taken according to (\ref{stable}),  serves as the Lyapunov function for the corresponding optimal Markov chain $\{X^*_t\}$.
\end{lemma}

\begin{proof}
Since $C > 0$, the function $c$ above satisfies
 $$\lim_{x\uparrow\infty}\min_{\nu}c(x, \nu) = \infty.$$
Thus it is `near-monotone' in the sense of (1.5), p.\ 58, \cite{Controlled_Markov}. The first claim then  follows by Theorem 1.1, p.\ 58, of \textit{ibid.} (This is proved under irreducibility assumption, but the same arguments work for unichain case, see, e.g., Lemmas 11.8 and 11.9, p.\ 353, \cite{Borkar} for an even more general result.) The second claim follows as in Lemma \ref{general}: in fact the counterpart of Lemma \ref{general} holds here with $W$ in place of $\Phi$.\footnote{We shall implicitly use this fact in what follows whenever we invoke Lemma \ref{general}.}
\end{proof}

\begin{lemma}
\label{lemma:Average_Cost}
$\lim_{\alpha\uparrow 1}\overline{V}_{\alpha} = V$ and $\lim_{\alpha\uparrow 1}(1 - \alpha)V_{\alpha}(0) = \beta$ where $(V, \beta)$ satisfy (\ref{DP2}). Furthermore, $\beta$ is uniquely characterized as the optimal cost $\beta(\lambda)$ and $V$ is rendered  unique on states that are positive recurrent under an optimal policy under the additional condition $V(0) = 0$\footnote{Without the latter, it is unique for positive recurrent states up to an additive constant.}. Finally, the Argmin of the right hand side in (\ref{DP2}) yields the optimal choice of $\nu$ for state $x$.
\end{lemma}

\begin{proof}
Let $\beta^*$ denote the supremum of average costs under all stationary policies, which is finite by Lemma \ref{general} (3). Let $v_{\alpha}$ denote an optimal stationary policy for the $\alpha$-discounted problem and $\{X_t\}$ a chain controlled by $v_{\alpha}$ with initial condition depending on the context. Let $\beta_{\alpha}$ denote the corresponding stationary expectation of the cost.  Then $\beta_{\alpha} \leq \beta^*$. Also,
\begin{equation}
|E\left[c(X_t, v_{\alpha}(X_t))| X_0 = i\right] - \beta_{\alpha}| \leq K\eta^n \label{expo}
\end{equation}
for some $K > 0, 0 < \eta < 1$ by Lemma \ref{general}. Thus
\begin{eqnarray}
\left|\overline{V}_{\alpha}(x)\right| &=& \Bigg|E\left[\sum_{t=0}^{\infty}\alpha^tc(X_t, v_{\alpha}(X_t)) \Big| X_0 = x\right] - \nonumber\\
&& E\left[\sum_{t=0}^{\infty}\alpha^tc(X_t, v_{\alpha}(X_t)) \Big| X_0 = 0\right]\Bigg| \nonumber  \\
&=& \Bigg|E\left[\sum_{t=0}^{\infty}\alpha^t(c(X_t, v_{\alpha}(X_t)) - \beta_{\alpha}) \Big| X_0 = x\right] - \nonumber\\
&& E\left[\sum_{t=0}^{\infty}\alpha^t(c(X_t, v_{\alpha}(X_t)) - \beta_{\alpha}) \Big| X_0 = 0\right]\Bigg|  \nonumber \\
&\leq& \sum_{t=0}^{\infty}\alpha^t \left|E\left[c(X_t, v_{\alpha}(X_t)) - \beta_{\alpha}) \Big| X_0 = x\right]\right| +  \nonumber \\
&& \sum_{t=0}^{\infty}\alpha^t \left|E\left[c(X_t, v_{\alpha}(X_t)) - \beta_{\alpha}) \Big| X_0 = 0\right]\right|  \nonumber \\
&\leq& \frac{\bar{K}}{1 - \eta} < \infty \label{bound1}
\end{eqnarray}
by (\ref{expdecay}) for suitable $\bar{K} > 0$. Similarly,
\begin{eqnarray*}
\left|(1 - \alpha)V_{\alpha}(0)\right| &=& (1 - \alpha)\left|E\left[\sum_{t=0}^{\infty}\alpha^tc(X_t, v_{\alpha}(X_t)) \Big| X_0 = 0\right]\right| \\
&\leq& (1 - \alpha)\sum_{t=0}^{\infty}\alpha^t\left|E\left[(c(X_t, v_{\alpha}(X_t)) - \beta_{\alpha})\Big| X_0 = 0\right]\right|  + \beta_{\alpha} \\
&\leq& \frac{(1 - \alpha)\bar{K}}{1 - \eta} + \beta_{\alpha}.
\end{eqnarray*}
Thus  by a standard Tauberian theorem \cite{Filar},
\begin{equation}
\limsup_{\alpha\uparrow 1}(1 - \alpha)V_{\alpha}(0) \leq \limsup_{\alpha\uparrow 1}\beta_{\alpha} \leq \beta^* < \infty.
\label{bound2}
\end{equation}
In view of (\ref{bound1})-(\ref{bound2}), we can invoke the Bolzano-Weirstrass theorem to let $\alpha\uparrow 1$ in (\ref{DPbar}) and conclude that any limit point $(V'(\cdot), \beta')$ of $(\overline{V}_{\alpha}(\cdot), (1 - \alpha)V_{\alpha}(0))$ as $\alpha\uparrow 1$ satisfies (\ref{DP2}). Clearly, $\overline{V}_{\alpha}(0) = 0$, hence $V'(0) = 0$. Furthermore, it follows from (\ref{expdecay}) and the arguments leading to (\ref{bound1}) that $V'$ will be $O(W( \cdot ))$. That $\beta' = \beta(\lambda)$ can be proved by a standard argument: We can drop the minimization on the right hand side of (\ref{DP2}) by replacing $p$ therein by $v'(x)$ where $v'(x)$ is the argmin. Then (\ref{DP2}) itself is a stochastic Lyapunov equation that proves the stability of $v'$. Taking expectations on both sides with respect to the stationary distribution under $v'$ yields, after some cancellations, that $\beta' =$ the stationary expectation of the cost under $v'$. On the other hand, if we replace $p$ by $v(x)$ on the right hand side for some other stable stationary policy $v$ and drop the minimization, we must replace the equality by $\leq$. Taking expectations on both sides with respect to the stationary expectation under $v$ then shows that $\beta' \leq$ the corresponding stationary cost. This proves that $v'$ is optimal and $\beta' = \beta(\lambda)$. The uniqueness part for $V$ can be proved by establishing an explicit representation as in Lemma 2.5, p.\ 79, \cite{Controlled_Markov}. We omit the details as this is not important for our purpose. The last statement is also standard, see section VI.2 of \textit{ibid.}
\end{proof}

\section{Structural properties of the value function}

In this section we prove the threshold nature of the optimal policy for the individual control problems through a sequence of lemmas.

\begin{lemma}\label{increase}
V is increasing in its argument.
\end{lemma}

\begin{proof}

Consider two processes $\{ X_t \}, \{ X_t' \}$ which follow equation (\ref{eq:dynamics}), with common arrival and control processes and with the initial conditions $y > z$ resp. Now, since the departure probability decreases as the number of jobs in the system increases, a stochastic dominance argument shows that $X_t > X_t'$ a.s. Since the cost function is an increasing function of the state variable, we have
$$J_\alpha(z,\{ \nu_t \} ) \leq J_\alpha(y,\{ \nu_t \} ).$$
Taking the minimum over all admissible policies on both sides, we get
$$V_{\alpha}(z) \leq V_{\alpha}(y),$$
implying
$$\bar{V}_{\alpha}(z) \leq \bar{V}_{\alpha}(y).$$
Letting $\alpha \uparrow 1$, and using Lemma \ref{lemma:Average_Cost}, we get $V(z) < V(y)$.

\ignore{Consider two processes $\{ X_t \}, \{ X_t' \}$ which follow equation (\ref{eq:dynamics}), with common arrival and control processes and common  departure mechanism as above, with the initial conditions $y > z$. Then $X_t > X_t'$ a.s.  Since the cost function is an increasing function of the state variable, we have
$$J_\alpha(z,\{ \nu_t \} ) \leq J_\alpha(y,\{ \nu_t \} ).$$
Taking the minimum over all admissible policies on both sides, we get
$$V_{\alpha}(z) \leq V_{\alpha}(y),$$
implying
$$\bar{V}_{\alpha}(z) \leq \bar{V}_{\alpha}(y).$$
Letting $\alpha \uparrow 1$, and using Lemma 3.2, we get $V(z) < V(y)$.}

\end{proof}
\begin{lemma}
\label{lemma:inc_diff}
V has increasing differences, i.e., $$z>0, x>y \Longrightarrow V(x+z) - V(x) \geq V(y+z) - V(y).$$
\end{lemma}
\begin{proof}
Consider the dynamic programming equation for the n-step finite horizon $\alpha$ discounted problem
\begin{equation}
V_t(x) = \min_{\nu} [c(x,\nu) + \alpha \int V_{t-1}(x-D+\nu\xi)\tilde{p}_1(dD | x)\tilde{p}_2(d\xi)] \: \:, \:t>0, \label{DPeq1}
\end{equation}
with $V_0(x) = Cx, x \geq 0$.
Here $x \mapsto \tilde{p}_1( \cdot | x)$ is the conditional distribution of departures given the current state, and $\tilde{p}_2$ is the arrival distribution. We embed the state space into $\mathbb{R}^+$, i.e., instead of $x$ taking only discrete values (which denote the number of jobs in a server), we allow x to take any real value in $[0, \infty)$. The departures $D$ as well as the arrivals $\xi$ take only non-negative integer values. However, the departure process depends on the number of jobs in the server x which is expected to be an integer. The departure process, which is determined by equation (\ref{Dep_Process}), will be modified in this case by replacing $x$ with $\lfloor x\rfloor$ when $x \geq 1$. For $x < 1$ there will be no departures ($\frac{y}{\lfloor 0 \rfloor} = 0$ by convention). Note that when $x$ is integral, the process will remain integral, recovering the original paradigm.
\newline
\newline
We first show that $V_t(x)$ is convex $\forall t$.
Define
$$h_t(x,\nu) = [c(x,\nu) + \alpha \int V_{t-1}(x-D+\nu\xi)\tilde{p}_1(dD|x)\tilde{p}_2(d\xi)] \: \:, \:n>0. $$
We use an induction argument similar to the one in \cite{Agarwal}.
$V_0(x) = Cx$ is a convex function. Assume that $V_{t-1}(x)$ is a convex function. Consider two points $x_1$ and $x_2$, such that $x_1>x_2$. Let the minimum on the right hand side of (\ref{DPeq1}) for $x = x_1, x_2$ be attained at $\nu_1$ and $\nu_2$ respectively. Then we have
 $$V_t(x_1) = c(x_1,\nu_1) + \alpha \int V_{t-1}(x_1-D+\nu_1\xi)\tilde{p}_1(dD|x_1)\tilde{p}_2(d\xi),$$
 $$V_t(x_2) = c(x_2,\nu_2) + \alpha \int V_{t-1}(x_2-D+\nu_2\xi)\tilde{p}_1(dD|x_2)\tilde{p}_2(d\xi).$$
Adding the two equations, we have
\begin{eqnarray*}
\lefteqn{V_t(x_1) + V_t(x_2) = c(x_1,\nu_1) + c(x_2,\nu_2)} \\
&&+ \ \alpha \int \tilde{p}_1(dD|x_1)\tilde{p}_2(d\xi)V_{t-1}(x_1-D+\nu_1\xi)\\
&&+ \  \alpha \int \tilde{p}_1(dD|x_2)\tilde{p}_2(d\xi) V_{t-1}(x_2-D+\nu_2\xi)\\
&\geq^{*^1}& 2\bigg[c((x_1+x_2)/ 2,(\nu_1+\nu_2)/2) +\\
&& \alpha \int \tilde{p}_1(dD | (x_1 + x_2)/2)\tilde{p}_2(d\xi)V_{t-1}((x_1+x_2)/ 2-D+((\nu_1+\nu_2)/2)\xi) \bigg]\\
&=& 2h_t((x_1+x_2)/ 2,(\nu_1+\nu_2)/2)\\
&\geq^{*^2}& 2V_t((x_1+x_2)/ 2).\\
\end{eqnarray*}

\noindent The first term on the right hand side of $(*^1)$ follows easily from the linearity of the cost function. The second term is derived as follows:
\begin{eqnarray*}
\lefteqn{\int \tilde{p}_1(dD|x_j)\tilde{p}_2(d\xi) V_{t-1}(x_j-D+\nu_1\xi) =} \\
&& \int_{D_1} \tilde{p}_1(dD|x_j)\tilde{p}_2(d\xi)V_{t-1}(x_j-D+\nu_1\xi) \\
&+& \int_{D_2} \tilde{p}_1(dD|x_j)\tilde{p}_2(d\xi) V_{t-1}(x_j-D+\nu_1\xi)
\end{eqnarray*}
for $j = 1,2,$ where
\begin{itemize}
\item $D_1$ = More than $\frac{x_1 + x_2}{2}$ departures,
\item $D_2$ = Less than (or equal to) $\frac{x_1 + x_2}{2}$ departures.
\end{itemize}
Since the departure probability is increased when going from $x_1$ to $x_2$ ($\frac{q}{\lfloor x_2\rfloor} \geq \frac{q}{\lfloor x_1\rfloor}$) and $V_{t-1}(.)$ is increasing in its argument by Lemma \ref{increase}, a stochastic dominance argument leads to
\begin{eqnarray*}
\lefteqn{\int_{D_2} \tilde{p}_1(dD|x_1)\tilde{p}_2(d\xi)V_{t-1}(x_1-D+\nu_1\xi) \geq } \\
&&\int_{D_2} \tilde{p}_1(dD|x_2)\tilde{p}_2(d\xi)V_{t-1}(x_1-D+\nu_1\xi).
\end{eqnarray*}
\newline
By the convexity of $V_{t-1}$, we have
\begin{eqnarray*}
&& \int_{D_2} \tilde{p}_1(dD|x_2)\tilde{p}_2(d\xi)V_{t-1}(x_1-D+\nu_1\xi) + \\
&& \ \ \ \ \ \ \ \ \ \ \ \ \ \ \ \ \ \int_{D_2} \tilde{p}_1(dD|x_2)\tilde{p}_2(d\xi)V_{t-1}(x_2-D+\nu_1\xi)\\
& \geq& 2 \bigg[ \int_{D_2} \tilde{p}_1(dD|x_2)\tilde{p}_2(d\xi) V_{t-1}((x_1+x_2)/ 2-D+\nu_1\xi) \bigg].
\end{eqnarray*}
Hence
\begin{eqnarray}
&& \int_{D_2} \tilde{p}_1(dD|x_1)\tilde{p}_2(d\xi)V_{t-1}(x_1-D+\nu_1\xi) + \nonumber \\
&& \ \ \ \ \ \ \ \ \ \ \ \ \ \ \ \ \  \int_{D_2} \tilde{p}_1(dD|x_2)\tilde{p}_2(d\xi)V_{t-1}(x_2-D+\nu_1\xi) \nonumber \\
& \geq 2& \bigg[ \int_{D_2} \tilde{p}_1(dD|x_2)\tilde{p}_2(d\xi) V_{t-1}((x_1+x_2)/ 2-D+\nu_1\xi) \bigg]. \label{bubba1}
\end{eqnarray}
Since
$$\int_{D_1} p_1(dD|(x_1+x_2)/ 2)p_2(d\xi) V_{n-1}((x_1+x_2)/ 2-D+\nu_1\xi) = 0,$$
$$\int_{D_1} p_1(dD|x_2)p_2(d\xi) V_{n-1}(x_2-D+\nu_1\xi) = 0,$$
and
$$\int_{D_1} \tilde{p}_1(dD|x_1)\tilde{p}_2(d\xi) V_{t-1}(x_1-D+\nu_1\xi) > 0,$$
we have
\begin{eqnarray}
\lefteqn{\int_{D_1} \tilde{p}_1(dD|x_2)\tilde{p}_2(d\xi)V_{t-1}(x_1-D+\nu_1\xi) + } \nonumber \\
&&\int_{D_1} \tilde{p}_1(dD|x_2)\tilde{p}_2(d\xi)V_{t-1}(x_2-D+\nu_1\xi) \nonumber \\
& \geq& 2 \bigg[ \int_{D_1} \tilde{p}_1(dD|x_2)\tilde{p}_2(d\xi) V_{t-1}((x_1+x_2)/ 2-D+\nu_1\xi) \bigg]. \label{bubba2}
\end{eqnarray}
Summing up the  inequalities (\ref{bubba1})-(\ref{bubba2}),  $(*^1)$ follows.
$(*^2)$ is due to the fact that $h_t(x,\nu) \geq V(x)$ $\forall x$.
This proves the convexity of the value function $V_t$ on  $\mathcal{R}^+$ $\forall t.$
This claim can be extended to infinite horizon discounted cost and then to average cost by standard limiting arguments (as in section 3 in the latter case) using the fact that pointwise limits of convex functions are convex. Thus  the average cost value function of our problem is convex. Convexity implies increasing differences. Therefore $V$ has increasing differences. Thus, the function restricted to non-negative integers will retain this property.
This proves that the original $V$ has increasing differences.
\end{proof}

\ \\

\begin{lemma}\label{thresholdpolicy}
The optimal policy is a threshold policy.
\end{lemma}

\ \\

The proof uses the structural properties of the value function established above and is very lengthy, therefore it is given separately in the Appendix.
The threshold nature of the policy has some important implications. Note that we have the unichain property: the existence of a reference state (in our case, $0$) that is accessible from all other states in the sense that the probability of eventually reaching it is strictly positive for any initial state. This implies in particular that there will be at most one communicating class. Under any stable stationary policy, there will be at least one communicating class, hence there has to be exactly one communicating class. There can, however, be transient states. By the threshold nature of the policy, we know that the set of active states will be of the form $A = \{0, 1, \cdots, k\}$ where $k$ is the threshold and the communicating class will be precisely $B = \{ 0, 1, \cdots, k, k+1\}$, because the chain can move from $k$ to $k+1$ when there is an arrival and no departure, but no arrival is allowed at state $k+1$ which is passive, so the chain has to return to $A$.  \\

 The uniform geometric ergodicity condition facilitated the vanishing discount argument of section 3 that led to a solution of the dynamic programming equation. Usually the values of the value function $V$ at transient states would be of little interest as these states do not affect the optimal cost, but in our case $V$ is an intermediary for our proof of Whittle indexability, so it continues to be of importance. \\

 Furthermore, we have considered the individual control problems for agents with a fixed common subsidy $\lambda$, which leads to a parametrized family of control problems. In particular, the threshold guaranteed by the above lemma depends on $\lambda$. The definition of Whittle indexability then requires that we establish the monotone increase of this threshold as $\lambda$ increases. This is what we do in the next section.\\

\section{Whittle indexability}

We are now ready to establish the Whittle indexability for our problem. We need the following intermediate result. Let $\pi^k$ denote the stationary distribution for the threshold policy when the threshold is $k$.\\

\begin{lemma}\label{stat}
$\sum_{j=0}^{k}\pi^{k}(j)$ is an increasing function of $k$.
\end{lemma}
\begin{proof}
This can be shown using the idea of stochastic dominance of Markov chains.
Consider two Markov chains, $\{X_t^1\} , \{ X_t^2 \} $ corresponding to the thresholds $k$ and $(k+1)$. Label the positive recurrent states  in the reverse order, i.e., if the threshold is $x$, state with $(x+1)$ jobs is labeled as state $0$, that with $x$ jobs is labeled as state $1$, and so on, till the state with $0$ jobs is labeled $(x + 1)$.
We will show that $\{ X_t^2 \} $ stochastically dominates $\{ X_t^1 \} $.\\

Let the transition matrices of the Markov chains be be $P_1, P_2$ respectively. Their elements $p_{\ell}(i,j), \ell = 1,2,$ give the corresponding probability of going from state $i$ (i.e., $(x+1-i)$ jobs in the server) to state $j$ (i.e., $(x+1-j)$ jobs in the server). These are distributed according to the stipulated departure process and the Bernoulli arrival Process. Extend $P_1$ to an $(k+2)\times(k+2)$ matrix by adding a column and row of zeros. Now consider the following $(k+2)\times(k+2)$ matrix:
\[
U=
  \begin{bmatrix}
    1 & 0 & 0 &. &. &. &. \\
    1 & 1 & 0 & 0 &. &. &. \\
    1 & 1 & 1 & 0 &0 &. &. \\
    . &. &. &. &. &. &.\\
    . &. &. &. &. &. &.\\
    1 &1 &1 &1 &1 &1 &1\\
  \end{bmatrix}
\]
We can easily see that $$P_1U \leq P_2U.$$
This shows that $\{ X_t^2 \} $ stochastically dominates $\{ X_t^1 \}$ - see Definition 3.12, Page 158, \cite{Kijima}.
Using Theorem 3.31, Page 158, \cite{Kijima}, we have  $$\hat{\pi}^\text{k+1}(0) \leq \hat{\pi}^\text{k}(0)$$
where $\hat{\pi}^k$ (resp., $\hat{\pi}^{k+1}$) is the stationary distribution for the threshold $k$ (resp., $k+1$) with the states relabeled as above. In view of our relabeling of the positive recurrent states, this translates into the following in the original notation:
$$\pi^\text{k+1}(k+2) \leq \pi^\text{k}(k+1).$$
This shows that the stationary probability of the only passive state with positive stationary probability decreases as the threshold increases.
However, $$\sum_{j=0}^{k+1}\pi^{k}(j) = 1$$ and  $\pi^{k}(k+1)$ decreases with the threshold $k$. This shows that $\sum_{j=0}^{k}\pi^{k}(j)$ is an increasing function of $k$.
\end{proof}
\begin{theorem}
This problem is Whittle indexable
\end{theorem}
\begin{proof}
The optimal average cost of the problem is given by
$$\beta(\lambda) = \inf \: \{ C\sum_kk\pi(k) + \lambda\sum_{k \in B}\pi(k) \}$$
where $\pi$ the stationary distribution and $B$ the set of passive states. $\beta(\lambda)$ is the infimum of this over all admissible policies, which by Lemma \ref{thresholdpolicy} is the same as the infimum over all threshold policies.  Hence it is concave non-decreasing with slope $<1$.
By Lemma \ref{thresholdpolicy}, we can write this as:
\begin{align*}
\beta(\lambda) &= C\sum_{k=0}^{\infty}k\pi^{x(\lambda)}(k) + \lambda\sum_{k=x(\lambda)+1}^{\infty}\pi^{x(\lambda)}(k)\\
&= C\sum_{k=0}^{\infty}k\pi^{x(\lambda)}(k) - \lambda\sum_{k=0}^{x(\lambda)}\pi^{x(\lambda)}(k) + \lambda\\
\end{align*}
where $x(\lambda)$ is the optimal threshold under $\lambda$. By the envelope theorem (Theorem 1, \cite{Milgrom}), the derivative of this function with respect to $\lambda$ is given by
$$-\sum_{k=0}^{x(\lambda)}\pi^{x(\lambda)}(k) + 1.$$
Since $\beta (\lambda)$ is a concave function, its derivative has to be a non-increasing function of $\lambda$, i.e.
$$\sum_{k=0}^{x(\lambda)}\pi^{x(\lambda)}(k) \: \: \: \text{is non-decreasing with} \: \lambda.$$
But, from Lemma \ref{stat}, we know that $\sum_{j=0}^{k}\pi^{k}(j)$ is a non-decreasing function of $k$, where $k$ is the threshold.
Therefore $x(\lambda)$ is a non-decreasing function of $\lambda$. The threshold for each $\lambda$ can be taken as the smallest value of the state $x$ for which it is equally favorable to stay active and passive. The set of passive states for $\lambda$ is given by
$ [x(\lambda), \infty)$.\
Since $x(\lambda)$ is an non-decreasing function of $\lambda$, we have that the set of passive states monotonically decreases to $\phi$ as $\lambda \uparrow \infty$. This implies Whittle indexability.
\end{proof}
\hfill\break
  The Whittle index policy at any time is to activate the server with the smallest value of Whittle Index for the current state profile. We do not have an explicit expression for the Whittle index, so have to take recourse to an iterative scheme to compute it. There is already a considerable body of work on computational schemes for Whittle index. For example,  \cite{Nino3} reviews a recursive algorithm to compute the Whittle index (or the Marginal Productivity Index (MPI)). Closed form expressions for the index have been obtained in some special cases, e.g., \cite{Nino1}. An iterative approach similar to the one presented below can be found in \cite{Glaze3}. A more general treatment for different computational approaches can be found in \cite{Nino4}. A linear programming approach to stochastic bandits which aims to improve the computational efficiency of calculating the index is derived in \cite{Nino4}.

 We have opted here for a rather simple scheme. Specifically, our algorithm computes the Whittle index for state $x$
by iteratively solving the following equation. (Here $\gamma$ is a small step size.)
$$\lambda_\text{t+1} = \lambda_\text{t} + \gamma(\sum_j p_a(j|x)V_{\lambda_\text{t}}(j) - \sum_j p_b(j|x)V_{\lambda_\text{t}}(j) - \lambda_\text{t}), \: \: \: t \geq 0$$
where $p_a(\cdot|.)$ is the transition probability when active and $p_b(\cdot|.)$ is the transition probability when the queue is passive.
\newline
\newline
This is simply an incremental scheme that adjusts the current guess for the index in the direction of decreasing the discrepancy in the active and  passive values which should agree for the exact index. The equations for $V_{\lambda}$ form a linear system which can be solved after each iteration using the current value of $\lambda$. That is, solve the following system of equations for $V = V_{\lambda_t}$ and $\beta = \beta(\lambda_t)$ using $\lambda = \lambda_t$. The combined scheme yields the Whittle index for a fixed $x$. We propose performing this computation for a sufficiently large number of $x$ followed by (say, linear) interpolation.

\begin{align*}
V(y) &= Cy -\beta + \sum_z p_a(z|y)V(z), \: \: y \leq x,\\
V(y) &= Cy + \lambda -\beta + \sum_z p_b(z|y)V(z), \: \: y > x,\\
V(0) &= 0.
\end{align*}

\section{Other Policies for Comparison}

\subsection{Exact Solution}
\label{section_lower_bound}

Consider a system with $I$ server and let $x_1, x_2, \cdots , x_I$ denote the number of jobs in these servers.
\newline
\newline
The Bellman equation for the system is given by:
\begin{align}
V(x_1,x_2, \cdots, x_I) + \beta = \sum_{i} c_ix_i + \min_{i} \big( \mathbb{E}^i[V(.)|x_1,x_2, \cdots, x_I] \big) \nonumber
\end{align}
Here $\mathbb{E}^i[ \ . \ ]$ denotes the expectation with respect to the transition probability where server $i$ is active and all others are passive, and $\beta$ represents the optimal cost of the problem. These equations face the curse of dimensionality and cannot be solved easily even for a moderate number of servers in the system. We show on how to evaluate these equations, using relative value iteration \cite{Puterman}, when there are two servers in the system (the case for more servers is similar).
\begin{align}
V_{n+1}(x_1,x_2) &= C_1x_1+C_2x_2 + \min_{i \in \{ 1,2 \}} \big( \mathbb{E}^i[V_n(.)|x_1,x_2] \big) - V_n(\overline{x_1}, \overline{x_2}), \nonumber \\
\nu_{n+1}(x_1,x_2) &= \underset{i}{\operatorname{argmin}}  \big( \mathbb{E}^i[V_n(.)|x_1,x_2] \big) \nonumber
\end{align}
Here $\nu(x_i,x_2)$ denotes the optimal server to activate when the system is in state $(x_1,x_2)$ and $(\overline{x_1}, \overline{x_2}) $ is any fixed state, say $(0,0)$.

\subsection{`$C \mu$' rule}
\label{section_Cmu}

Another heuristic that we consider as a candidate for scheduling  is inspired by the $C \mu$ rule \cite{Nain}.
This rule assigns priority to the queue with the highest value of $C \mu$ (or the lowest value of $\frac{1}{C \mu}$) where C is the holding cost per job in the queue and $1/\mu$ is the service time requirement of the job. We have suitably changed this rule to match our setting
where there are multiple servers, each serving at different rates. The service rate varies from server to server and also changes with the number of jobs present in each server (Processor Sharing). This gives the average processing rate at server $i$ as $q_i/X_t^i$ where $X_t^i$ is the number of jobs at time $t$ in server $i$. This suggests the following index for each server. Choose the server which has the lowest value of `cost per unit rate' given by
$$\frac{C_i X_t^i}{q_i}$$
where $C_i$ is the cost of holding a job in server $i$.
Here, we have used the reciprocal `cost per unit rate' instead of `rate per unit cost' in order to avoid division by 0. This is a simple greedy heuristic that favors higher service rate while penalizing higher cost. By abuse of terminology, we shall refer to this as the $C \mu$ rule, though the framework here is quite distinct from that of the classical $C \mu$ rule.

\subsection{Random Allocation}
\label{section_random}

In this scheme, at each time instant we pick the server to be activated with uniform distribution independently of all else. The rates of the servers, the number of jobs in the server or the cost of service are not considered in making the choice.

\section{Numerical Results}

In Figures \ref{fig:model_bound1} and \ref{fig:model_bound2}, we see the suboptimality gap of the proposed Whittle index policy with the solution of the original problem, as described in section \ref{section_lower_bound}. The solution to the original problem faces the curse of dimensionality and it becomes infeasible to compute as the number of servers increases. We have performed the simulations for two servers in the system and the figures show that the Whittle policy (which is asymptotically optimal in the $I\uparrow\infty$ limit) is not far away from the optimal even for the case of two servers (i.e., $I = 2$).

\begin{figure}[ht]
  \includegraphics[scale=0.6]{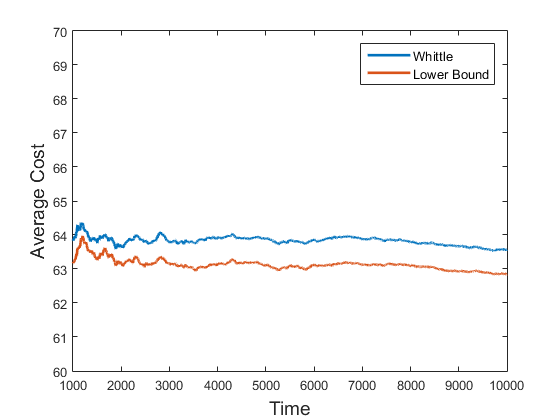}
  \caption{Cost are 100 and 90 with service capacities 0.55 and 0.50 respectively}
  \label{fig:model_bound1}
\end{figure}

\begin{figure}[ht]
  \includegraphics[scale=0.6]{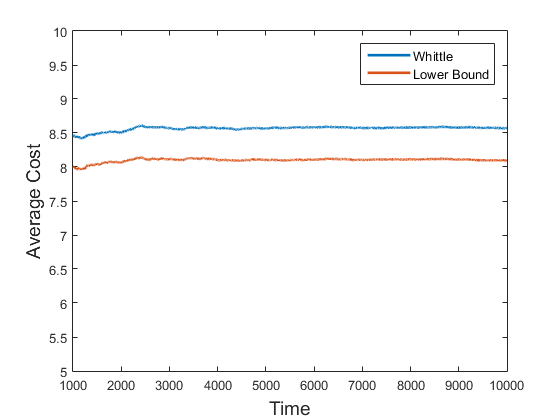}
  \caption{Cost are 12 and 11 with service capacities 0.55 and 0.45 respectively}
  \label{fig:model_bound2}
\end{figure}

In figures \ref{fig:model1}, \ref{fig:model2}, \ref{fig:model3} and \ref{fig:model4} we compare three schemes: a version of the  `$C \mu$ rule' as described in section \ref{section_Cmu}, random server allocation policy as described in section \ref{section_random}, and the Whittle index policy.
\newline
\newline
For these simulations, we consider a system where there are three servers and there is a single arrival process which is Bernoulli. While our foregoing analysis is for countably infinite state space, we used a buffer size of 100 for our numerical experiments. The input arrival rate is fixed at $p = 0.4$. Again, as the system is finite, we do not require the assumption $q_I > 2p$, as mentioned before. The three policies are compared based on different processing rates and different holding costs.
\newline
\newline
While using the Whittle Index policy, we did not ever have to look at the Whittle index of higher states (states above 40). Since the probability of arrivals is significantly less that of departure, higher states have a negligible probability of occurrence and are practically unseen in simulations.

\begin{figure}[H]
  \includegraphics[scale=0.6]{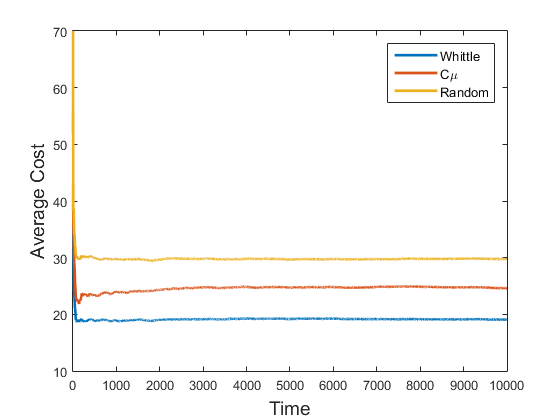}
  \caption{Cost are 30, 29, 28 with service capacities 0.55, 0.50, 0.45 respectively}
  \label{fig:model1}
\end{figure}

\begin{figure}[H]
  \includegraphics[scale=0.6]{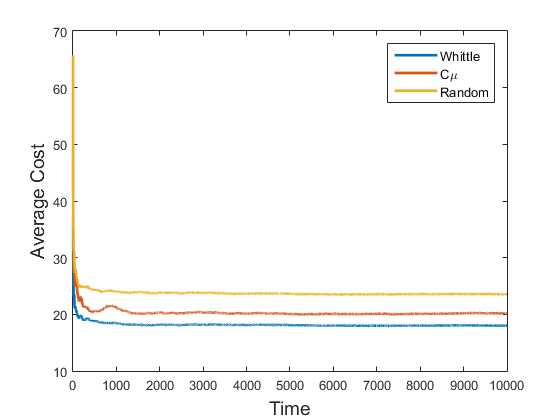}
  \caption{Cost are 30, 29, 28 with service capacities 0.95, 0.50, 0.45 respectively}
  \label{fig:model2}
\end{figure}

\begin{figure}[H]
  \includegraphics[scale=0.6]{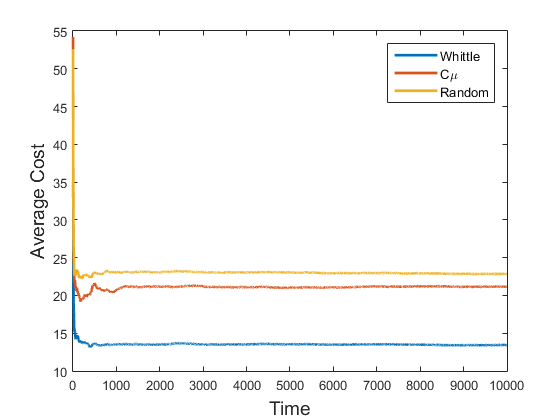}
  \caption{Cost are 40, 23, 16 with service capacities 0.55, 0.50, 0.45 respectively}
  \label{fig:model3}
\end{figure}

\begin{figure}[H]
  \includegraphics[scale=0.6]{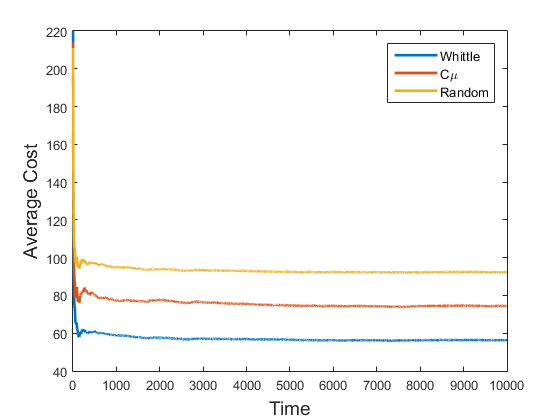}
  \caption{Cost are 100, 90, 80 with service capacities 0.55, 0.50, 0.45 respectively}
  \label{fig:model4}
\end{figure}

The numerical experiments show a clear advantage for the Whittle  index policy.\\

\section{Conclusions and Future Work}
We have considered a simple processor sharing model and established its Whittle indexability. Several possible extensions suggest themselves. On the technical side, ensuring uniform geometric ergodicity was a great boon and facilitated many proofs. In fact, we needed the strong condition $q_I > 2p$ that seemed to arise because of the service rate diminishing with queue length. It would be good to replace it by a  weaker assumption. Ideally, the fact that the cost is `near-monotone' in the sense of  \cite{Controlled_Markov}, Chapter V-VI, and therefore penalizes unstable behavior, should itself suffice for stability of the Whittle index policy. To prove this remains a technical challenge.  Also, there are more complicated models of processor sharing and it would be interesting to see to what extent this line of thinking can be generalized. Finally, the computation of Whittle index is a major computational saving on the computation of the true optimal policy, nevertheless it is still computationally intensive. It would be good to come up with provably good approximations with accompanying error bounds. One such attempt based on reinforcement learning ideas appears in \cite{Chadha}.

\section{Appendix}

\noindent \textbf{Proof of Lemma \ref{thresholdpolicy}:}\\

Recall that $\xi$ is Bernoulli($p$), $D$ is the departure random variable which is distributed as $D \sim\ Bin(x,\frac{q}{x})$ where $x$ is the number of jobs currently in the server and $q$ is the processing rate of the server. Consider the following:
\newline
\begin{align*}
f(x) &= \mathbb{E}_x[V(x - D + \xi )] - \mathbb{E}_x[V(x - D)] \\
&= \sum_{d=0}^{x}{x \choose d}\bigg(\frac{q}{x} \bigg)^d \bigg(1-\frac{q}{x} \bigg)^{x-d}(pV(x+1-d) + (1-p)V(x-d) - V(x-d))\\
&= \sum_{d=0}^{x}{x \choose d}\bigg(\frac{q}{x} \bigg)^d \bigg(1-\frac{q}{x} \bigg)^{x-d}(p(V(x+1-d) - V(x-d)))\\
\end{align*}
For an optimal policy to be a threshold policy, we need $f(x+d) \geq f(x)$ for $d>0$. Therefore it is sufficient to show that $f(x+1) - f(x) > 0$. We have:
\begin{align*}
f(x+1) &= \sum_{d=0}^{x+1}{x+1 \choose d}\bigg(\frac{q}{x+1} \bigg)^d \bigg(1-\frac{q}{x+1} \bigg)^\text{x+1-d}\times\\
& \ \ \ \ \ \ \ \ \ \ \ (p(V(x+2-d) - V(x+1-d))).
\end{align*}
Consider the following bound for $f(x)$:
\begin{align*}
f(x) & \leq \sum_{d=0}^{x}{x \choose d}\bigg(\frac{q}{x+1} \bigg)^d \bigg(1-\frac{q}{x+1} \bigg)^{x-d}(p(V(x+1-d) - V(x-d))).
\end{align*}
This is true because the departure probability  stochastically decreases as $x$ increases.  Since $(V(x+1-d) - V(x-d))$ decreases as d increases by Lemma \ref{lemma:inc_diff}, the bound follows by stochastic dominance.
\newline
\newline
Now we look at the difference
\begin{align*}
f(x+1) - f(x) &= \sum_{d=0}^{x}{x+1 \Big[\choose d+1}\bigg(\frac{q}{x+1} \bigg)^\text{d+1} \bigg(1-\frac{q}{x+1} \bigg)^{x-d}\times\\
& \ \ \ \ \ \ \ \ \ \ \ p(V(x+1-d) - V(x-d))\\
&- {x \choose d}\bigg(\frac{q}{x} \bigg)^d \bigg(1-\frac{q}{x} \bigg)^{x-d}p(V(x+1-d) - V(x-d))\Big] \\
&+ \bigg(1 - \frac{q}{x+1} \bigg)^\text{x+1} p(V(x+2) -V(x+1))\\
& \geq \bigg(1 - \frac{q}{x+1} \bigg)^\text{x+1} p(V(x+2) -V(x+1))\\
&+ p\sum_{d=0}^{x}{x \choose d} \bigg( \frac{q}{x+1} \bigg)^d \bigg(1-\frac{q}{x+1} \bigg)^{x-d}\times\\
& \ \ \ \ \ \ \ \ \ \ \ \bigg(\frac{x+1}{d+1} \frac{q}{x+1} -1\bigg)(V(x+1-d) - V(x-d)). \\
\end{align*}
Let $q_x = \frac{q}{x+1}$. Then we have
\begin{align*}
f(x+1) - f(x) & \geq \big(1 - q_x \big)^\text{x+1} p (V(x+2) -V(x+1))\\
&+ p\sum_{d=0}^{x}{x \choose d} \big( q_x \big)^d \big(1-q_x \big)^{x-d}\times\\
& \ \ \ \ \ \ \ \ \ \ \  \bigg(\frac{x+1}{d+1} q_x -1\bigg)(V(x+1-d) - V(x-d)). \\
\end{align*}
Let $B(x) = p(V(x+1) - V(x))$.
From the Lemma \ref{lemma:inc_diff}, we know that B(.) is a non-decreasing function.
The above inequality can be rewritten as
\begin{align*}
f(x+1) - f(x) & \geq \big(1 - q_x \big)^\text{x+1}B(x+1)\\
&+ \sum_{d=0}^{x}{x \choose d} \big( q_x \big)^d \big(1-q_x \big)^{x-d} \bigg(\frac{x+1}{d+1} q_x -1\bigg)B(x-d). \\
\end{align*}
We simplify the expression as follows:
\begin{align}
&f(x+1) - f(x) \nonumber \\
&\geq \big(1 - q_x \big)^\text{x+1}B(x+1) + \sum_{d=0}^{x}{x \choose d} \big( q_x \big)^d \big(1-q_x \big)^{x-d} \times \nonumber \\
& \ \ \ \ \ \ \ \ \ \ \ \bigg(\frac{x+1}{d+1} q_x -1\bigg)B(x-d) \nonumber \\
&= (1-q_x)^{x+1}  B(x+1) - q_x^x(1-q_x)B(0) \nonumber \\
&+ \sum_{d=0}^{x-1} {x \choose d} q_x^{d} (1-q_x)^{x-d} \left(  \frac{x+1}{d+1}q_x - 1 \right)B(x-d) \label{three}
\end{align}
\begin{align}
& = (1-q_x)^{x+1}  B(x+1) - q_x^x(1-q)B(0) \ + \nonumber \\
& \ \ \ \ \ \ \ \ \ \ \ \ \sum_{d=0}^{x-1} \left({x \choose d+1} q_x^{d+1} (1-q_x)^{x-d} - {x \choose d} q_x^{d} (1-q_x)^{x+1-d} \right)B(x-d) \label{four}\\
& = \left\{ (1-q_x)^{x+1}  B(x+1) + \sum_{d=0}^{x-1}{x \choose d+1} q_x^{d+1} (1-q_x)^{x-d}B(x-d)\right\} - \nonumber \\
& \ \ \ \ \ \ \ \ \ \ \ \ \ \left\{ q_x^x(1-q_x)B(0)  + \sum_{d=0}^{x-1}{x \choose d} q_x^{d} (1-q_x)^{x+1-d}B(x-d) \right\} \nonumber \\
& = \left\{(1-q_x)^{x+1}  B(x+1) + \sum_{d=1}^{x}{x \choose d} q_x^{d} (1-q_x)^{x+1-d}B(x+1-d)\right\} \nonumber \\
&- \sum_{d=0}^{x}{x \choose d} q_x^{d} (1-q_x)^{x+1-d}B(x-d)\nonumber \\
& = \sum_{d=0}^{x}{x \choose d} q_x^{d} (1-q_x)^{x+1-d}B(x+1-d) - \sum_{d=0}^{x}{x \choose d} q_x^{d} (1-q_x)^{x+1-d}B(x-d) \nonumber \\
& = (1-q_x) \sum_{d=0}^{x}{x \choose d} q_x^{d} (1-q_x)^{x-d} \left(B(x+1-d) - B(x-d) \right) \nonumber \\
& \geq 0. \nonumber
\end{align}
The passage from (\ref{three}) to (\ref{four}) is given in \cite{Cloud}. The proof has been reproduced below for sake of completeness.
We have proved that
$$f(x+1) - f(x) \geq 0,$$
which shows that the optimal policy is a threshold policy.

\bigskip

\noindent \textbf{Proof of term in (\ref{three}) to (\ref{four}) }

From (\ref{three}), we have:

\begin{align*}
{x \choose d} q_x^{d} (1-q_x)^{x-d} \left(  \frac{(x+1)q_x}{d+1} - 1 \right)
& = {x \choose d} q_x^{d} (1-q_x)^{x-d} \left(  \frac{(x+1)q_x \pm q_xd}{d+1} - 1 \right) \\
& = {x \choose d} q_x^{d} (1-q_x)^{x-d} \left(  \frac{(x-d)q_x + (d+1)q_x}{d+1} - 1 \right) \\
& = {x \choose d} q_x^{d} (1-q_x)^{x-d} \left(  \frac{(x-d)q_x}{d+1} - 1 + q_x \right) \\
& = {x \choose d} q_x^{d} (1-q_x)^{x-d} \left(  \frac{(x-d)q_x}{d+1} - (1 - q_x) \right) \\
& = {x \choose d} q_x^{d} (1-q_x)^{x-d} \frac{(x-d)q_x}{d+1} \\
& \quad - {x \choose d} q_x^{d} (1-q_x)^{x-d}(1 - q_x) \\
& = \frac{x!}{(x-d)! d!} \frac{(x-d)}{d+1} q_x^{d+1} (1-q_x)^{x-d}  \\
& \quad - {x \choose d} q_x^{d} (1-q_x)^{x+1-d} \\
& = \frac{x!}{(x-d-1)! (d+1)!} q_x^{d+1} (1-q_x)^{x-d}  \\
& \quad - {x \choose d} q_x^{d} (1-q_x)^{x+1-d} \\
& = {x \choose d+1} q_x^{d+1} (1-q_x)^{x-d} \\
& \quad - {x \choose d} q_x^{d} (1-q_x)^{x+1-d} \\
\end{align*}
The term on the right hand side of the final step is the desired term in (\ref{four}).


\begin{thebibliography}{}


\bibitem{Aalto} Aalto, Samuli; Ayesta, Urtzi; Borst, Sem; Misra, Vishal and N\'{u}\"{n}ez-Queija, Rudesindo. ``Beyond processor sharing." ACM SIGMETRICS Performance Evaluation Review 34.4 (2007): 36-43.

\bibitem{Agarwal} Agarwal, Mukul; Borkar, Vivek S.\ and Karandikar, Abhay. ``Structural properties of optimal transmission policies over a randomly varying channel." IEEE Transactions on Automatic Control 53.6 (2008): 1476-1491.

\bibitem{Altman2} Altman, Eitan; Avrachenkov, Konstantin and Ayesta, Urtzi. ``A survey on discriminatory processor sharing." Queueing systems 53.1-2 (2006): 53-63.

\bibitem{Altman1} Altman, Eitan; Ayesta, Urtzi and Prabhu, Balakrishna J. ``Load balancing in processor sharing systems." Telecommunication Systems 47.1-2 (2011): 35-48.

\bibitem{Glaze1} Ansell, P.\ S.; Glazebrook, Kevin D.\ and Kirkbride, C.\ ``Generalized `Join the shortest queue' policies for the dynamic routing of jobs for multi-class queues." Journal of the Operational Research Society 54.4 (2003):379-389.

\bibitem{Glaze4} Archibald, T.\ W.; Black, D.\ P.\ and Glazebrook, Kevin D.\ ``Indexability and index heuristics for a simple class of inventory routing problems." Operations Research 57.2 (2009): 314-326.

\bibitem{Glaze2} Argon, Nilay Tanik; Ding, Li; Glazebrook, Kevin D.\ and Zia, Serhan, ``Dynamic routing of customers with general delay costs in a multiserver queuing system." Probability in the Engineering and Informational Sciences 23.12 (2009): 175-203.

\bibitem{Bertsekas} Bertsekas, Dimitri P.\ \textit{Dynamic programming and optimal control}. Vol. 1, Belmont, MA: Athena Scientific, 1995.

\bibitem{Controlled_Markov} Borkar, Vivek S.\ \textit{Topics in controlled Markov chains}. Harlow, UK:Longman Scientific $\&$ Technical, 1991.

\bibitem{Borkar} Borkar, Vivek S.\ ``Convex analytic methods in Markov decision processes." \textit{Handbook of Markov decision processes} (A.\ Shwartz and E.\ Feinberg, eds.), Norwell, MA: Kluwer Academic, 2002. 347-375.

\bibitem{Chadha} Borkar, Vivek S.\ and Chadha, Karan, ``A reinforcement learning algorithm for restless bandits", \textit{submitted}.

\bibitem{Cloud} Borkar, Vivek S., Karumanchi, Ravikumar and Saboo, Krishnakant. \textit{``An index policy for dynamic pricing in cloud computing under price commitments"}, Applicationes Mathematicae, to appear.

\bibitem{Borkar_2} Borkar, Vivek S.\ ``Uniform stability of controlled Markov processes." \textit{System Theory: Modeling, Analysis and Control} (T.\ E.\ Djaferis and I.\
 C.\ Schick, eds.)  Norwell, MA: Kluwer Academic, 2000. 107-120.

\bibitem{time_scale} Borkar, Vivek S.\ ``Stochastic approximation with two time scales." Systems \& Control Letters 29.5 (1997): 291-294.

    \bibitem{Glaze3} Glazebrook, Kevin D.; Kirkbride, C.\ and Ouenniche, J.\ ``Index policies for the admission control and routing of impatient customers to homogeneous service stations." Operations Research 57.4 (2009): 975-989.

\bibitem{Jacko}       Jacko, Peter, \textit{Dynamic priority allocation in restless bandit models}, Lambert Academic Publishing, 2010.

\bibitem{Kijima} Kijima, Masaaki. \textit{Markov processes for stochastic modeling}.  Springer Science + Business Media, 1997.

\bibitem{Kleinrock} Kleinrock, Leonard. ``Time-shared systems: a theoretical treatment." Journal of the ACM  14.2 (1967): 242-261.

\bibitem{Meyn} Meyn, Sean P.\ and Tweedie, Richard L.\ \textit{Markov chains and stochastic stability} (2nd ed.) Cambridge, UK: Cambridge Uni.\ Press, 2012.

\bibitem{Larranga} Larranaga, Maialen; Ayesta, Urtzi and Verloop, Ina Maria. ``Dynamic control of birth-and-death restless bandits: application to resource-allocation problems." availabe at\\
https://hal.archives-ouvertes.fr/hal-01064370/file/LAV2015$\%$20(1).pdf

\bibitem{Milgrom} Milgrom, Paul and  Segal, Ilya. ``Envelope theorems for arbitrary choice sets." Econometrica 70.2 (2002): 583-601.

\bibitem{Nain} Nain, Philippe. ``Interchange arguments for classical scheduling problems in queues." Systems $\&$ control letters 12.2 (1989): 177-184.

\bibitem{Nain2} Nain, Philippe; Tsoucas, Pantelis  and Walrand, Jean. ``Interchange arguments in stochastic scheduling." Journal of Applied Probability (1989): 815-826.

\bibitem{Nino1} Nin\~{o}-Mora, Jos\`{e}, ``Towards minimum loss job routing to parallel heterogeneous multiserver queues via index policies"' European Journal of Operational Research 220.3 (2012): 705-720.

\bibitem{Nino2} Nin\~{o}-Mora, Jos\`{e}, ``Admission and routing of soft real-time jobs to multi-clusters: Design and comparison of index policies." Computers $\&$ Operations Research 39.12 (2012): 3431-3444.

\bibitem{Nino3} Nin\~{o}-Mora, Jos\`{e}, ``Marginal productivity index policies for admission control and routing to parallel multi-server loss queues with reneging." Proceedings of the International Conference on Network Control and Optimization, June 5-7, 2007,  Avignon, France, Lecture Notes in Computer Science No.\ 4465 (Tijani Chahed and Bruno Tuffin, eds.), Springer Berlin Heidelberg (2007): 138-149.

\bibitem{Nino4} Nin\~{o}-Mora, Jos\`{e}. ``Characterization and computation of restless bandit marginal productivity indices." Proceedings of the 2nd international conference on Performance evaluation methodologies and tools. ICST (Institute for Computer Sciences, Social-Informatics and Telecommunications Engineering), 2007.

\bibitem{Papa} Papadimitriou, C.\ H.\ and  Tsitsiklis, J.\ N.\  ``The complexity of optimal queuing network control", \textit{Mathematics of Operations Research} 24.2 (1999), 293-305.

\bibitem{Puterman} Puterman, Martin L.\ \textit{Markov decision processes: discrete stochastic dynamic programming}. New York: John Wiley $\&$ Sons, 2014.


\bibitem{Filar} Sznajder, R.\ and Filar, Jerzy A.\  ``Some comments on a theorem of Hardy and Littlewood." \textit{J.\ Optimization Theory and Applications} 75.1 (1992), 201-208.

\bibitem{Whittle} Whittle, Peter ``Restless bandits: Activity allocation in a changing world." Journal of Applied Probability 25 (1988): 287-298.


\end{thebibliography}
\end{document}